\documentclass[letterpaper, 12pt, conference]{ieeeconf}  % Comment this line out if you need a4paper

\IEEEoverridecommandlockouts                              % This command is only needed if
\usepackage{amsmath} % assumes amsmath package installed
\usepackage{chngcntr}
\usepackage{amssymb}  % assumes amsmath package installed
\usepackage{graphicx} % for pdf, bitmapped graphics files
\usepackage{mathptmx} % assumes new font selection scheme installed
\usepackage{times} % assumes new font selection scheme installed
\usepackage{datetime}
\usepackage{enumerate}
\usepackage{arydshln} % dashlines in matrices
\newtheorem{lem}{Lemma}
\newtheorem{thm}{Theorem}

\newtheorem{definition}{Definition}

\def\<{\leqslant}           % nice less than or equal to sign
\def\>{\geqslant}           % nice larger than or equal to sign
\def\wh{\widehat}
\def\wt{\widetilde}
\def\~{\wt{~}}

\def\Re{\mathrm{Re}}   % real part
\def\cH{\mathcal{H}}   % Hardy space
\def\mR{\mathbb{R}}    % real line
\def\mC{\mathbb{C}}    % complex plane
\def\fH{\mathfrak{H}}   % initial system space

\def\Tr{\mathrm{Tr}}       % matrix trace
\def\rT{\mathrm{T}}        % matrix transpose
\def\[[[{[\![\![}
\def\]]]{]\!]\!]}

\def\bra{{\langle}}
\def\ket{{\rangle}}

\def\rd{{\rm d}}        % differential

\def\fa{\mathfrak{a}}

\def\x{\times}
\def\ox{\otimes}

\def\fF{\mathfrak{F}}

\def\mD{\mathbb{D}}

\def\mbT{\mathbfit{T}}

\def\cR{\mathcal{R}}
\def\cRH{\mathcal{RH}}

\def\cP{{\mathcal P}}
\def\cQ{\mathcal{Q}}

\def\cA{\mathcal{A}}

\def\cS{{\mathcal S}}

\DeclareMathAlphabet      {\mathbfit}{OML}{cmm}{b}{it}
\DeclareMathAlphabet      {\mathbfd}{OT1}{cmr}{bx}{n}
\title{\LARGE\bf Parameterization of Stabilizing Linear Coherent Quantum Controllers}
\author{Arash Kh. Sichani, \qquad Ian R. Petersen, \qquad Igor G. Vladimirov%\\ \today, \currenttime\\
\thanks{This work is supported by the Australian Research Council and Air Force Office of Scientific Research. The authors are with UNSW Canberra, Australia: {\tt arash\_kho@hotmail.com, i.r.petersen@gmail.com, igor.g.vladimirov@gmail.com}.}
}
\onecolumn
\begin{document}
\maketitle
\thispagestyle{empty}
\pagestyle{plain}

%%%%%%%%%%%%%%%%%%%%%%%%%%%%%%%%%%%%%%%%%%%%%%%%%%%%%%%%%%%%%%%%%%%%%%%%%%%%%%%%%%%%%%%%%%%%%%%%%%%
\begin{abstract}
    This paper is concerned with application of the classical Youla-Ku\v{c}era parameterization to finding a set of linear coherent quantum controllers that stabilize a linear quantum plant. The plant and controller are assumed to represent open quantum harmonic oscillators modelled by linear quantum stochastic differential equations. The interconnections between the plant and the controller are assumed to be established through  quantum bosonic fields. In this framework, conditions for the stabilization of a given linear quantum plant via linear coherent quantum feedback are addressed using a stable factorization approach. The class of stabilizing quantum controllers is parameterized in the frequency domain. Also, this approach is used in order to formulate coherent quantum weighted $\cH_2$ and $\cH_\infty$ control problems for linear quantum systems in the frequency domain. Finally, a projected gradient descent  scheme is proposed to solve the coherent quantum weighted $\cH_2$ control problem.
\end{abstract}

%%%%%%%%%%%%%%%%%%%%%%%%%%%%%%%%%%%%%%%%%%%%%%%%%%%%%%%%%%%%%%%%%%%%%%%%%%%%%%%%

\section{INTRODUCTION}\label{sec:Intro}
%%%%%%%%%%%%%%%%%%%%%%%%%%%%%%%%%%%%%%%%%%%%%%%%%%%%%%%%%%%%%%%%%%%%%%%%%%%%%%%%

    Coherent quantum feedback control builds on the idea of constructing a feedback loop from the interconnection of quantum systems through field coupling \cite{L_2000}. This technique avoids loss of quantum information in conversion to classical signals  which occurs during measurement, as a direct consequence of the projection postulate in quantum mechanics \cite{M_1998}. The coherent quantum control approach aims at developing systematic methods to design measurement-free interconnections of quantum systems modelled by quantum stochastic differential equations (QSDEs); see for example \cite{JNP_2008,NJP_2009,P_2010}.
    %%%%%%%%%%%%%%%%%%%%%%%%%%%%%%
    Quantum-optical components, such as optical cavities, beam splitters and phase shifters, make it possible to implement quantum feedback systems governed by linear QSDEs \cite{P_1992,M_2008,P_2010}, provided the latter represent open quantum harmonic oscillators \cite{EB_2005,GZ_2004}. This important class of linear QSDEs models the Heisenberg evolution of pairs of conjugate operators in a multi-mode quantum  harmonic oscillator that is coupled to external bosonic fields. The condition of physical realizability (PR) of a linear QSDE as an open quantum harmonic oscillator is organised as a set of constraints on the coefficients of the QSDE \cite{JNP_2008} or, alternatively, on the quantum system transfer matrix \cite{SP_2012} in the frequency domain.
    %%%%%%%%%%%%%%%%%%%%%%%%%%%%%%
    Coherent quantum feedback control problems, such as stabilization and robust controller design, are of particular interest in linear quantum control theory \cite{JNP_2008,P_2010}. These problems are amenable to transfer matrix design methods \cite{yan2003I, yan2003II,GJN10, P_2010}. There are classical approaches to control synthesis problems for linear multivariable systems based on the transfer matrix of the system \cite{ZDG_1996}. In particular, one of the important strategies in controller design for such systems is the stable factorization approach \cite{vid2011}.
    %%%%%%%%%%%%%%%%%%%%%%%%%%%%%%
    The central idea of the factorization approach is to represent the transfer matrix of a system as a ratio of stable rational matrices. This idea gives rise to a methodology which leads to the solution of several important control problems; see \cite{vid2011}. One of the fundamental results used in the factorization approach to classical control is the parameterization of all stabilizing controllers, which is known as the Youla-Ku\v{c}era parameterization.
    %%%%%%%%%%%%%%%%%%%%%%%%%%%%%%
    The Youla-Ku\v{c}era parameterization was developed originally in the frequency domain for finite-dimensional linear time-invariant systems using transfer function methods, see \cite{youla76I,youla76II}, and generalized to infinite-dimensional systems afterwards \cite{desoer80,quadrat2003,vid2011}. The state space representation of all stabilizing controllers has also been addressed for finite-dimensional, linear time-invariant \cite{nett84} and time-varying \cite{dale93} systems. Furthermore, the approach was shown to be applicable to a class of nonlinear systems \cite{hammeri85,paice90,anderson98}.
    %%%%%%%%%%%%%%%%%%%%%%%%%%%%%%
    In the present paper, we employ the stable factorization approach in order to develop a quantum counterpart of the classical Youla-Ku\v{c}era parameterization for describing a set of linear coherent quantum controllers that stabilize a linear quantum system. In particular, we address the problem of coherent quantum stabilizability of a given linear quantum plant. The class of stabilizing controllers is parameterized in the frequency domain. This approach allows weighted $\cH_2$ and $\cH_\infty$ coherent quantum control problems to be formulated for linear quantum systems in the frequency domain. In this way, the weighted $\cH_2$ and $\cH_\infty$ control problems are reduced to constrained optimization problems with respect to the Youla-Ku\v{c}era parameter with convex cost functionals. Moreover, these problems are organised as a constrained version of the model matching problem \cite{francis87}. Finally, a projected gradient descent scheme is proposed to solve the weighted $\cH_2$ coherent quantum control problem in the frequency domain.
    %%%%%%%%%%%%%%%%%%%%%%%%%%%%%%
%    The rest of this paper is organised as follows. Section~\ref{sec:not} outlines the notation used in the paper. The class of linear quantum systems under consideration is described in Section~\ref{sec:system} which also revisits the PR conditions for linear quantum systems in the frequency domain. Section~\ref{sec:closedsys} describes the coherent quantum feedback interconnection.
%    Sections~\ref{sec:cntlpara} and~\ref{sec:QYK} formulate the quantum version of the Youla-Ku\v{c}era parameterization and provide relevant preparatory material. Coherent quantum weighted $\mathcal{H}_2$ and $\mathcal{H}_\infty$ control problems are introduced in Section~\ref{sec:H2}. A projected gradient descent scheme for the quantum weighted $\mathcal{H}_2$ control problem is outlined in Section~\ref{sec:PGS}. Section~\ref{sec:Conclusion} gives concluding remarks. %Appendices~\ref{app:Appndx_LFT} and \ref{app:GBI} provide  subsidiary material on linear fractional transformations and the general B\'{e}zout identity.
%%\ref{app:Appndx_LFT}    \ref{app:GBI}

%%%%%%%%%%%%%%%%%%%%%%%%%%%%%%%%%%%%%%%%%%%%%%%%%%%%%%%%%%%%%%%%%%%%%%%%%%%%%%%%
\section{NOTATION}\label{sec:not}
%%%%%%%%%%%%%%%%%%%%%%%%%%%%%%%%%%%%%%%%%%%%%%%%%%%%%%%%%%%%%%%%%%%%%%%%%%%%%%%%

Vectors are assumed to be organised as columns unless specified otherwise, and the transpose $(\cdot)^{\rT}$ acts on matrices with operator-valued entries as if the latter were scalars. For a vector $a$ of operators $a_1, \ldots, a_r$ and a vector $b$ of operators $b_1, \ldots, b_s$, the commutator matrix
$
    [a,b^{\rT}]
    :=
    ab^{\rT} - (ba^{\rT})^{\rT}
$
is an $(r\x s)$-matrix whose $(j,k)$th entry is the commutator
$
    [a_j,b_k]
    :=
    a_jb_k - b_ka_j
$ of the operators $a_j$ and $b_k$.
Furthermore, $(\cdot)^{\dagger}:= ((\cdot)^{\#})^{\rT}$ denotes the transpose of the entry-wise operator adjoint $(\cdot)^{\#}$. When it is applied to complex matrices,  $(\cdot)^{\dagger}$ reduces to the complex conjugate transpose  $(\cdot)^*:= (\overline{(\cdot)})^{\rT}$. Also, $I_r$ denotes the identity matrix of order $r$, and
$
J_r:=
{\small\begin{bmatrix}
    I_r & 0\\
    0 & -I_r
\end{bmatrix}}
$ is a signature matrix.
%The spectral radius of a square matrix is denoted by $\br(\cdot)$.
The %adjoints of linear operators acting on  matrices are understood in the sense of
Frobenius inner product of real or complex matrices is denoted by
$
    \bra M,N\ket_{\rm F}
    :=
    \Tr(M^*N)
$ and generates the Frobenius norm $\|\cdot\|_{\rm F}$. Matrices of the form
${\small\begin{bmatrix}
    R_1 & R_2 \\
    \overline{R}_2 & \overline{R}_1
\end{bmatrix}}
$
are denoted by $\Delta(R_1,R_2)$. The imaginary unit is denoted by $i:= \sqrt{-1}$, and the $(j,k)$th block of a matrix $\Gamma$ is referred to as $\Gamma_{jk}$. %$\mC[s]$ denotes the set of polynomials in the indeterminate $s$ with coefficients in the field $\mC$ of complex numbers and degree defined in the standard way. The field of fractions associated with $\mC[s]$ is denoted by $\mC(s)$, and consists of rational functions in $s$ with complex coefficients. $\mS$ denotes the subset of $\mC(s)$ consisting of all proper stable rational functions.
The notation
$
    {\small\left[
    \begin{array}{c|c}
          A & B \\
          \hline
          C & D
    \end{array}
    \right]}
$
refers to a state space realization of the corresponding transfer matrix $\Gamma(s) := C(sI-A)^{-1}B+D$ with a complex variable $s \in \mC$. The conjugate system transfer matrix $(\Gamma(-\overline{s}))^*$ is written as $\Gamma^{\~}(s)$. The Hardy space of (rational) transfer functions of type $p=2,\infty$ is denoted by $\cH_p$ (respectively, $\cRH_p$). The symbol $\ox$ is used for the tensor product of spaces. %The Riccati operator $X=Ric(H)$ of a Hamiltonian matrix $H:=\begin{bmatrix} A & R \\ -Q & -A^* \end{bmatrix}$ returns $X$ which is real symmetric, satisfies the algebraic Riccati equation $A^*X+XA+XRX+Q=0$ such that $A+RX$ is stable.

%%%%%%%%%%%%%%%%%%%%%%%%%%%%%%%%%%%%%%%%%%%%%%%%%%%%%%%%%%%%%%%%%%%%%%%%%%%%%%%%%%%%%%%%%%%%%%%%%%%
\section{LINEAR QUANTUM STOCHASTIC SYSTEMS}\label{sec:system}
%%%%%%%%%%%%%%%%%%%%%%%%%%%%%%%%%%%%%%%%%%%%%%%%%%%%%%%%%%%%%%%%%%%%%%%%%%%%%%%%%%%%%%%%%%%%%%%%%%%

The open quantum systems under consideration are governed by linear QSDEs which model the dynamics of open quantum harmonic oscillators. The input-output maps of such systems can be described in the frequency domain by using the transfer function approach  \cite{yan2003I,yan2003II}. We will now outline this framework which is used as a basis of the frequency domain synthesis approach to quantum control presented in this paper.

%%%%%%%%%%%%%%%%%%%%%%%%%%%%%%%%%%%%%%%%%%%%%%%%%%%%%%%%%%%%%%%%%%%%%%%%%%%%%%%%
\subsection{Open Quantum Harmonic Oscillators}
%%%%%%%%%%%%%%%%%%%%%%%%%%%%%%%%%%%%%%%%%%%%%%%%%%%%%%%%%%%%%%%%%%%%%%%%%%%%%%%%

 Corresponding to a model of $n$ independent quantum harmonic oscillators is a vector $\fa$ of annihilation operators $\fa_1, \ldots, \fa_n$ on Hilbert spaces $\fH_1, \ldots, \fH_n$. The adjoint $\fa_j^\dagger$ of the operator $\fa_j$ is referred to as the creation operator. The doubled-up vector $\breve{\fa}$ of the annihilation and creation operators satisfies the canonical commutation relations (CCRs) \cite{M_1998}
\begin{equation}
\label{aaa}
    [
        \breve{\fa},
        \breve{\fa}^{\dagger}
    ]
    :=
    {\small\begin{bmatrix}
        [\fa,\fa^{\dagger}] & [\fa,\fa^{\rT}]\\
        [\fa^{\#},\fa^{\dagger}] & [\fa^{\#},\fa^{\rT}]
    \end{bmatrix}}
    = J_n,
    \qquad
    \breve{\fa}:=
    {\small\begin{bmatrix}
        \fa\\ \fa^{\#}\end{bmatrix}}.
\end{equation}
We consider a linear quantum system whose dynamic variables are linear combinations of the annihilation and creation operators, acting on the tensor product space $\fH:= \fH_1\ox \ldots \ox \fH_n$:
\begin{equation}
\label{FFa}
    a
    :=
    F_1 \fa+ F_2 \fa^{\# }
    =
    \begin{bmatrix}
        F_1 & F_2
    \end{bmatrix}
    \breve{\fa},
\end{equation}
where $F_1$ and $F_2$ are appropriately dimensioned complex matrices.
The relations (\ref{aaa}) and (\ref{FFa}) imply that
$$
    [\breve{a}, \breve{a}^{\dagger}]
    =
    F
    [\breve{\fa}, \breve{\fa}^{\dagger}]
    F^*
    =
    FJ_n F^*
    =:
    \Theta,
$$
where $F:= \Delta(F_1,F_2) \in \mC^{2n \times 2n}$ in accordance with the doubled-up notation \cite{GJN10},  and the complex Hermitian matrix $\Theta$ of order $2n$  is the generalized CCR  matrix \cite{SP_2012}.
Now, consider an $n$-mode open quantum harmonic oscillator interacting with an external bosonic field defined on a Fock space \cite{P_1992}. The oscillator is assumed to be coupled to $m$ independent external input bosonic fields acting on the tensor product space $\fF:=\fF_1 \ox\ldots \ox\fF_m$, where $\fF_j$ denotes the Fock space associated with the $j$th input channel. The field annihilation operators $\cA_1(t), \ldots, \cA_m(t)$, which act on $\fF$, form a vector $\cA_{\rm in}(t)$. Their adjoints  $\cA_1^{\dagger}(t), \ldots, \cA_m^{\dagger}(t)$, that is, the field creation operators, comprise a vector $\cA_{\rm in}^\#(t)$. The field annihilation and creation operators are adapted to the Fock filtration and satisfy the It\^{o} table $
    \rd
    \breve{\cA}_{\rm in}(t)
    \rd
    \breve{\cA}_{\rm in}^{\dagger}(t)
    =
    {\small\begin{bmatrix}
        I_m & 0\\
        0 & 0
    \end{bmatrix}}
    \rd t
$ in terms of the corresponding doubled-up vector $\breve{\cA}_{\rm in}(t):=   {\small\begin{bmatrix}
        \cA_{\rm in}(t)\\
        \cA_{\rm in}^{\#}(t)
    \end{bmatrix}}$.
The linear QSDEs, derived from the joint evolution of the $n$-mode open quantum harmonic oscillator and the external bosonic fields in the Heisenberg picture, can be represented in the following form \cite{SP_2012,GJN10}:
\begin{align}
\label{equ:tdomain_model:1}
    \rd
    \breve{a}(t)
    &= A \breve{a}(t) \rd t+
        B \rd \breve{\cA}_{\rm in}(t),\\
\label{equ:tdomain_model:2}
    \rd
    \breve{\cA}_{\rm out}(t)
    &= C \breve{a}(t)\rd t+
   			 D \rd \breve{\cA}_{\rm in}(t).
\end{align}
Here, the first QSDE governs the plant dynamics, while the second QSDE describes the dynamics of the output fields in terms of the corresponding doubled-up vector
$
\breve{\cA}_{\rm out}(t)
    :=
    {\small \begin{bmatrix}
        \cA_{\rm out}(t)\\
        \cA_{\rm out}^{\#}(t)
    \end{bmatrix}}
$
of annihilation and creation operators acting on the system-field composite   space $\fH\ox \fF$. Also, the matrices $A \in \mC^{2n\x 2n}$, $B \in \mC^{2n\x 2m}$, $C \in \mC^{2m\x 2n}$, $D \in \mC^{2m\x 2m}$ in (\ref{equ:tdomain_model:1}) and (\ref{equ:tdomain_model:2}) are given by
\begin{align}
    \label{equ:ABCD}
    {\small\begin{bmatrix}
        A & B\\
        C & D
    \end{bmatrix}}
    :=
    {\small\begin{bmatrix}
        -i\Theta H - \frac{1}{2} \Theta L^*J_m L & -\Theta L^* J_m \Delta(S,0)\\
        L & \Delta(S,0)
    \end{bmatrix}},
%
%    A &:= -i\Theta H - \frac{1}{2} \Theta L^*J_m L,
%    \qquad
%    %\label{equ:ABCD:B}
%    B := -\Theta L^* J_m \Delta(S,0),\\
%    %\label{equ:ABCD:C}
%    C &:= L,
%    \qquad\qquad \qquad \qquad \quad
%    \label{equ:ABCD:D}
%    D := \Delta(S,0)\!,
\end{align}
where $H=H^*=\Delta(H_1,H_2) \in \mC^{2n \times 2n}$ is a Hermitian matrix which parameterizes the system Hamiltonian operator $\frac{1}{2}\breve{a}^{\dagger} H\breve{a}$, the matrix $L=\Delta(L_1,L_2) \in \mC^{2m \times 2n}$ specifies the system-field coupling operators, and $S \in \mC^{m \times m}$ is the unitary scattering matrix.

%%%%%%%%%%%%%%%%%%%%%%%%%%%%%%%%%%%%%%%%%%%%%%%%%%%%%%%%%%%%%%%%%%%%%%%%%%%%%%%%
\subsection{Open Quantum Harmonic Oscillator in the Frequency Domain and Physical Realizability}\label{subsec:LQHO}
%%%%%%%%%%%%%%%%%%%%%%%%%%%%%%%%%%%%%%%%%%%%%%%%%%%%%%%%%%%%%%%%%%%%%%%%%%%%%%%%

The input-output map of the open quantum harmonic oscillator, governed by the linear QSDEs (\ref{equ:tdomain_model:1}) and (\ref{equ:tdomain_model:2}), is completely specified by the transfer matrix \cite{yan2003I,yan2003II,P_2010,GJN10} which is defined in the standard way  as
    \begin{equation}
    \label{equ:hoc_freq}
        \Gamma(s) :=
    {\small\left[
    \begin{array}{c|c}
          A & B\\
          \hline
          C & D
    \end{array}
    \right]},
    \end{equation}
where the matrices $A, B, C, D$ are given by the $(S,L,H)$-parameterization (\ref{equ:ABCD}).
In view of the specific structure of this parameterization, not every linear QSDE, or the system transfer matrix (\ref{equ:hoc_freq}) with an arbitrary quadruple $(A,B,C,D)$, represents the dynamics of an open quantum harmonic oscillator. This fact is addressed in the form of PR conditions for the quadruple $(A,B,C,D)$ to represent such an oscillator; see \cite{JNP_2008,NJP_2009} for more details.

For the purpose of the present paper, it is convenient to take advantage of the frequency domain version of the PR constraints on the system transfer matrices. For its formulation, we will need an auxiliary notion.

%%%%%%%%%%%%%%%%%%%%%%%%%%%%%%%%%%%%%%%%%%%%%%%%%%%%%%%%%%%%%%%%%%%%%%%%%%%%%%%%
\begin{definition}
\label{def:gen}
The matrix $A$ and the state-space realization (\ref{equ:hoc_freq}) are said to be \emph{spectrally generic} if  the spectrum $\sigma(A)$ has no intersection with its mirror reflection
about the imaginary axis in the complex plane:
$
    \sigma(A)\bigcap \left(-\overline{\sigma(A)} \right) = \emptyset
$,
that is, $\lambda+\overline{\nu} \ne 0$ for all eigenvalues $\lambda,\nu\in \sigma(A)$. \hfill$\square$
\end{definition}
%%%%%%%%%%%%%%%%%%%%%%%%%%%%%%%%%%%%%%%%%%%%%%%%%%%%%%%%%%%%%%%%%%%%%%%%%%%%%%%%

In particular, a spectrally generic matrix $A$ can not have purely imaginary eigenvalues.
Also, we will use a special class of complex matrices of order $2m$:
\begin{equation}
\label{mD}
    \mD_m := \big\{\Delta(S,0):\ S \in \mC^{m\x m}\ {\rm is\ unitary}\big\}.
\end{equation}%\vskip-2mm\noindent

%The following lemma provides necessary and sufficient conditions for PR in the frequency domain.
%%%%%%%%%%%%%%%%%%%%%%%%%%%%%%%%%%%%%%%%%%%%%%%%%%%%%%%%%%%%%%%%%%%%%%%%%%%%%%%%
\begin{lem}
\label{lem:PR_Freq}
\cite[Theorem~4, p.~2040]{SP_2012}
Suppose $\Gamma$ is a square transfer matrix of order $2m$ with a spectrally generic minimal state-space realization (\ref{equ:hoc_freq}).
Then $\Gamma$ represents an open quantum harmonic oscillator in the frequency domain if and only if
	\begin{equation}
		\label{equ:JJUnit}
		\Gamma^{\~}(s)J_m\Gamma(s)=J_m
	\end{equation}
for all $s\in \mC$,
	and
	the feedthrough matrix $D = \Gamma(\infty)$ belongs to the set $\mD_m$ in (\ref{mD}). \hfill $\square$
\end{lem}
%%%%%%%%%%%%%%%%%%%%%%%%%%%%%%%%%%%%%%%%%%%%%%%%%%%%%%%%%%%%%%%%%%%%%%%%%%%%%%%%

A transfer function $\Gamma$, satisfying the condition (\ref{equ:JJUnit}), is said to be $(J_m,J_m)$-unitary; see, for example, \cite{kim97} and references therein. Since we consider this property for square transfer matrices (in which case, (\ref{equ:JJUnit}) implies that $|\det \Gamma(i\omega)|=1$, and hence, $\Gamma(i\omega)$ is nonsingular for any $\omega \in \mR$), $(J_m,J_m)$-unitarity is equivalent to its dual form \cite{kim97,P_2010}:
\begin{equation}
\label{dual}
		\Gamma(s)J_m \Gamma^{\~}(s)=J_m.
\end{equation}
In view of (\ref{equ:JJUnit}) and (\ref{dual}),  the feedthrough matrix $D$ in Lemma~\ref{lem:PR_Freq} inherits $(J_m,J_m)$-unitarity $DJ_mD^* = D^*J_m D=J_m$  from the transfer function $\Gamma$ by continuity. Moreover,  for an arbitrary matrix $S\in \mC^{m\x m}$,  the matrix $D=\Delta(S,0)$ is $(J_m,J_m)$-unitary if and only if $S$ is unitary.
In what follows, the argument $s$ of transfer functions will often be omitted for brevity.

%%%%%%%%%%%%%%%%%%%%%%%%%%%%%%%%%%%%%%%%%%%%%%%%%%%%%%%%%%%%%%%%%%%%%%%%%%%%%%%%
\section{FEEDBACK INTERCONNECTION}\label{sec:closedsys}
%%%%%%%%%%%%%%%%%%%%%%%%%%%%%%%%%%%%%%%%%%%%%%%%%%%%%%%%%%%%%%%%%%%%%%%%%%%%%%%%

We will now consider a linear quantum plant and a linear quantum controller with square transfer matrices $P$ and $K$, respectively, each representing an open quantum harmonic oscillator in the frequency domain.
By analogy with similar structures in classical control settings, we partition the vectors $\cA_{\rm Pin}$ and $\cA_{\rm Pout}$ of the plant input and output field annihilation operator processes in accordance with Fig.~\ref{fig:modifiedsystem}:
\begin{equation}
\label{AA}
		\cA_{\rm Pin}=
        {\small\begin{bmatrix}
			\cA_r\\
			\cA_u
		\end{bmatrix}},
    \qquad
		\cA_{\rm Pout}=
		{\small\begin{bmatrix}
			\cA_z\\
			\cA_y
		\end{bmatrix}}.
\end{equation}
Here $\cA_r$, $\cA_z$, $\cA_y$, $\cA_u$ denote the vectors of annihilation operators of the input and output fields of the closed-loop system, and the input and output of the controller, which correspond to the classical reference, output, observation and control signals, respectively.
In order to bring the feedback interconnection to a standard format, the plant transfer matrix $P$ is modified to $\cP$ as
\begin{align}
		\nonumber
        \cP
        &~=
        {\small \left [
        \begin{array} {c c}
            \cP_{11} & \cP_{12}\\
            \cP_{21}&\cP_{22}
        \end{array}
        \right ]} \nonumber
        :=
        {\small\left[
        \begin{array}{c c : c c}
            P_{11} & P_{13} &  P_{12} & P_{14}\\
            P_{31} & P_{33} &  P_{32} & P_{34}\\
            \hdashline
            P_{21} & P_{23} &  P_{22} & P_{24}\\
            P_{41} & P_{43} &  P_{42} & P_{44}
        \end{array}	
        \right]}\\
%        \nonumber\\
        \nonumber
        &~=
    {\small\left [
    \begin{array}{c | c : c}
        A    & \begin{bmatrix} B_1 & B_3 \end{bmatrix} & \begin{bmatrix} B_2 & B_4 \end{bmatrix}\\[1mm]
        \hline\\[-2.5mm]
        \begin{bmatrix} C_1 \\ C_3 \end{bmatrix} & \begin{bmatrix} D_{11} & 0 \\
        	0 & \overline{D}_{11} \end{bmatrix}  & \begin{bmatrix} D_{12} & 0 \\
        	 0 & \overline{D}_{12} \end{bmatrix}  \\ [3.5mm]
		\hdashline \\ [-2.5mm]
        \begin{bmatrix} C_2 \\ C_4 \end{bmatrix} & \begin{bmatrix} D_{21} & 0 \\ 0 & \overline{D}_{21} 	\end{bmatrix}  & \begin{bmatrix} D_{22} & 0 \\ 0 & \overline{D}_{22} \end{bmatrix}
    \end{array}
    \right ]}
    \label{equ:modplnt} \\
    &=:\,
    {\small\left [
    \begin{array}{c|c c}
          A & \mathbfit{B}_1 & \mathbfit{B}_2\\
          \hline
          \mathbfit{C}_1 & \mathbfit{D}_{11} & \mathbfit{D}_{12} \\
          \mathbfit{C}_2 & \mathbfit{D}_{21} & \mathbfit{D}_{22}
    \end{array}
    \right]}.
\end{align}
The interconnection of the modified plant and the controller is shown in Fig.~\ref{fig:system}.
The closed-loop transfer matrix between the exogenous inputs and outputs of interest can be calculated through the lower linear fractional transformation (LFT) of the modified plant and the controller in the frequency domain \cite{ZDG_1996,GJN10}:
\begin{equation}
\label{G}
    G
    =
    \cP_{11} + \cP_{12}K(I-\cP_{22}K)^{-1}\cP_{21}
	=:
    {\rm LFT}(\cP,K).
\end{equation}
%==============================================================================
Note that, similarly to the classical case, the interconnection in Fig.~\ref{fig:system}
%==============================================================================
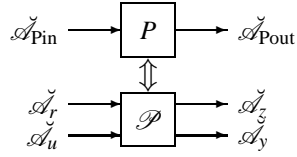
\begin{figure}[htpb]
\begin{center}
\unitlength=0.7mm
\linethickness{0.4pt}
\begin{picture}(30.00,40.00)
    %\put(7.5,7.5){\dashbox(35,15)[cc]{}}
    \put(10,10){\framebox(10,10)[cc]{\small$\cP$}}
    \put(10,27){\framebox(10,10)[cc]{\small$P$}}
    \put(15,23.5){\makebox(0,0)[cc]{{$\Updownarrow$}}}
    \put(-2,32){\makebox(0,0)[rc]{{\small$\breve{\cA}_{\rm Pin}$}}}
    \put(32,32){\makebox(0,0)[lc]{{\small$\breve{\cA}_{\rm Pout}$}}}
    \put(0,32){\vector(1,0){10}}
    \put(20,32){\vector(1,0){10}}
%    \put(15,10){\vector(0,1){10}}
% \put(15,8){\makebox(0,0)[ct]{{$\omega$}}}
%    \put(5,15){\line(0,-1){10}}
%    \put(5,5){\line(1,0){40}}
%    \put(45,5){\line(0,1){10}}
 %   \put(50,30){\line(1,0){10}}
 %   \put(35,25){\vector(0,-1){5}}
    %\put(30,18){\vector(-1,0){10}}
%    \put(20,28){\line(1,0){10}}
%    \put(30,28){\line(0,-1){13}}
    \put(0,12){\vector(1,0){10}}
    \put(0,18){\vector(1,0){10}}
    \put(20,12){\vector(1,0){10}}
    \put(20,18){\vector(1,0){10}}
%    \put(10,15){\line(-1,0){10}}
%    \put(0,15){\line(0,1){13}}
%    \put(15,25){\vector(0,-1){5}}
    \put(-2,12){\makebox(0,0)[rc]{{\small$\breve{\cA}_u$}}}
    \put(32,12){\makebox(0,0)[lc]{{\small$\breve{\cA}_y$}}}
    \put(-2,18){\makebox(0,0)[rc]{{\small$\breve{\cA}_r$}}}
    \put(32,18){\makebox(0,0)[lc]{{\small$\breve{\cA}_z$}}}
\end{picture}\vskip-5mm
\caption{This diagram depicts the way in which the original plant $P$ is modified to $\cP$ by partitioning the vectors of the plant input and output field operators in (\ref{AA}). This modified structure allows for the connection to another linear quantum system which acts as the controller.}
\label{fig:modifiedsystem}
\end{center}
\end{figure}
%==============================================================================
\begin{figure}[htpb]
\begin{center}
\unitlength=0.7mm
\linethickness{0.4pt}
\begin{picture}(30.00,35.00)
    %\put(7.5,7.5){\dashbox(35,15)[cc]{}}
    \put(10,25.5){\framebox(10,10)[cc]{\small$\cP$}}
    \put(10,10){\framebox(10,10)[cc]{\small$K$}}
%    \put(15,10){\vector(0,1){10}}
% \put(15,8){\makebox(0,0)[ct]{{$\omega$}}}
    \put(0,33){\vector(1,0){10}}
    \put(20,33){\vector(1,0){10}}
    \put(0,28){\vector(1,0){10}}
%    \put(5,15){\line(0,-1){10}}
%    \put(5,5){\line(1,0){40}}
%    \put(45,5){\line(0,1){10}}
 %   \put(50,30){\line(1,0){10}}
 %   \put(35,25){\vector(0,-1){5}}
    %\put(30,18){\vector(-1,0){10}}
    \put(20,28){\line(1,0){10}}
    \put(30,28){\line(0,-1){13}}
    \put(30,15){\vector(-1,0){10}}
    \put(10,15){\line(-1,0){10}}
    \put(0,15){\line(0,1){13}}
%    \put(15,25){\vector(0,-1){5}}
    \put(-2,15){\makebox(0,0)[rc]{{\small$\breve{\cA}_u$}}}
    \put(32,15){\makebox(0,0)[lc]{{\small$\breve{\cA}_y$}}}
    \put(-2,33){\makebox(0,0)[rc]{{\small$\breve{\cA}_r$}}}
    \put(32,33){\makebox(0,0)[lc]{{\small$\breve{\cA}_z$}}}
\end{picture}\vskip-5mm
\caption{This diagram depicts the fully quantum closed-loop system  which is the interconnection of the modified quantum plant $\cP$ and the quantum controller $K$. The effect of the environment on the closed-loop system is represented by $\breve{\cA}_r$.}
\label{fig:system}
\end{center}
\end{figure}
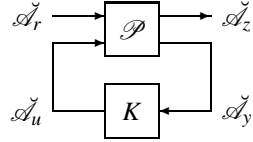
%==============================================================================
provides a general framework for the feedback interconnection of two quantum systems, one of which acts as the plant and the other as the controller. This framework includes the conventional coherent quantum feedback interconnection shown in Fig.~\ref{fig:routing}.
%==============================================================================
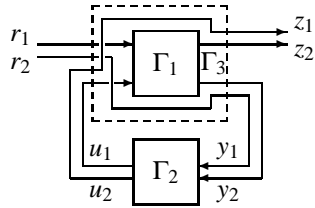
\begin{figure}[htpb]
\begin{center}
\unitlength=0.85mm
\linethickness{0.4pt}
\begin{picture}(30.00,45.00)
    \put(3.5,23){\dashbox(21,18)[cc]{}}
    \put(10,27){\framebox(10,10)[cc]{\small$\Gamma_1$}}
    \put(17.5,27){\makebox(10,10)[cc]{\small$\Gamma_3$}}
    \put(10,10){\framebox(10,10)[cc]{\small$\Gamma_2$}}
%    \put(15,10){\vector(0,1){10}}
% \put(15,8){\makebox(0,0)[ct]{{$\omega$}}}
    \put(5.5,35){\vector(1,0){4.5}}
    \put(-5,35){\line(1,0){9.5}}

    \put(-5,33){\line(1,0){9.5}}
    \put(5.5,33){\line(1,0){1}}
    \put(6.5,33){\line(0,-1){8}}
    \put(6.5,25){\line(1,0){15.5}}
    \put(22,25){\line(0,1){2}}
    \put(22,27){\line(1,0){6}}
    \put(28,27){\line(0,-1){11}}

    \put(20,35){\vector(1,0){14}}
    \put(7,29){\vector(1,0){3}}
%    \put(5,15){\line(0,-1){10}}
%    \put(5,5){\line(1,0){40}}
%    \put(45,5){\line(0,1){10}}
 %   \put(50,30){\line(1,0){10}}
 %   \put(35,25){\vector(0,-1){5}}
    %\put(30,18){\vector(-1,0){10}}
    \put(20,29){\line(1,0){10}}
    \put(30,29){\line(0,-1){15}}
    \put(28,16){\vector(-1,0){8}}
    \put(30,14){\vector(-1,0){10}}
    \put(10,14){\line(-1,0){10}}
    \put(10,16){\line(-1,0){8}}
    \put(2,16){\line(0,1){13}}
    \put(2,29){\line(1,0){4}}

    \put(0,14){\line(0,1){17}}
    \put(0,31){\line(1,0){5}}
    \put(5,31){\line(0,1){8}}
    \put(5,39){\line(1,0){17}}
    \put(22,39){\line(0,-1){2}}
    \put(22,37){\vector(1,0){12}}

%    \put(15,25){\vector(0,-1){5}}
    \put(23,17){\makebox(0,0)[lb]{{\small$y_1$}}}
    \put(23,13){\makebox(0,0)[lt]{{\small$y_2$}}}
    \put(3,17){\makebox(0,0)[lb]{{\small$u_1$}}}
    \put(3,13){\makebox(0,0)[lt]{{\small$u_2$}}}

    \put(-6,35){\makebox(0,0)[rb]{{\small$r_1$}}}
    \put(-6,33){\makebox(0,0)[rt]{{\small$r_2$}}}
    \put(35,37){\makebox(0,0)[lb]{{\small$z_1$}}}
    \put(35,35){\makebox(0,0)[lt]{{\small$z_2$}}}
\end{picture}\vskip-5mm
\caption{This diagram depicts the way in which the quantum system $\Gamma_3$, the concatenation of $\Gamma_1$ and the feedthroughs, is formed by grouping the exogenous inputs to and outputs from the closed-loop system.}
\label{fig:routing}
\end{center}
\end{figure}
%==============================================================================
%\begin{figure}[htpb]
%\begin{center}
%%\includegraphics[scale=.9]{routing}
%\caption{This diagram depicts the way in which the quantum system $\Gamma_3$, the concatenation of $\Gamma_1$ and the feedthroughs, is formed by grouping the exogenous inputs to and outputs from the closed-loop system.}
%\label{fig:routing}
%\end{center}
%\end{figure}
%%%%%%%%%%%%%%%%%%%%%%%%%%%%%%%%%%%%%%%%%%%%%%%%%%%%%%%%%%%%%%%%%%%%%%%%%%%%%%%%%%%%%
In the latter figure, the exogenous inputs and outputs of the closed-loop system are grouped together. Here, $\Gamma_3$ (which is the concatenation of $\Gamma_1$ and the feedthroughs) and $\Gamma_2$ represent the transfer matrices $\cP$ and $K$ of the modified quantum plant and the quantum controller, respectively.
Also, note that the PR conditions are usually formulated for the case when the number of exogenous inputs to the closed-loop system is not less than the number of outputs of the controller \cite{JNP_2008,NJP_2009}.
%%%%%%%%%%%%%%%%%%%%%%%%%%%%%%%%%%%%%%%%%%%%%%%%%%%%%%%%%%%%%%%%%%%%%%%%%%%%%%%%%%%%%%%%%%%%%%%%%%%%

%%%%%%%%%%%%%%%%%%%%%%%%%%%%%%%%%%%%%%%%%%%%%%%%%%%%%%%%%%%%%%%%%%%%%%%%%%%%%%%%%%%%
\section{PARAMETERIZATIONS OF STABILIZING CONTROLLERS}\label{sec:cntlpara}
%%%%%%%%%%%%%%%%%%%%%%%%%%%%%%%%%%%%%%%%%%%%%%%%%%%%%%%%%%%%%%%%%%%%%%%%%%%%%%%%%%%%

For the purposes of Section~\ref{sec:QYK}, we will now briefly review the Youla-Ku\v{c}era parameterization of classical stabilizing controllers together with related notions. The latter include stabilizability, detectability, internal stability, coprime factorizations and matrix fractional descriptions (MFDs). Despite the quantum control context, these notions will be used according to their standard definitions in classical linear control theory \cite{ZDG_1996,vid2011}.

%%%%%%%%%%%%%%%%%%%%%%%%%%%%%%%%%%%%%%%%%%%%%%%%%%%%%%%%%%%%%%%%%%%%%%%%%%%%%%%%%%%%
\subsection{Stabilizability of Feedback Connections}
%%%%%%%%%%%%%%%%%%%%%%%%%%%%%%%%%%%%%%%%%%%%%%%%%%%%%%%%%%%%%%%%%%%%%%%%%%%%%%%%%%%%

Consider the two-two block of the modified plant transfer matrix $\cP$ in (\ref{equ:modplnt}) given by
\begin{equation}
    \label{equ:G22}
    \cP_{22} :=
        {\small\begin{bmatrix}
            P_{22} & P_{24} \\
            P_{42} & P_{44}
        \end{bmatrix}}
    =
    {\small\left[
        \begin{array}{c | c}
            A               & \mathbfit{B}_2 \\
            \hline
            \mathbfit{C}_2  & \mathbfit{D}_{22}
        \end{array}
    \right]}.
\end{equation}
The following lemma provides a necessary and sufficient condition for the internal stability of the feedback system in Fig.~\ref{fig:system}.

\begin{lem} \label{lem:equi} \cite[Lemma~12.2, p.~294]{ZDG_1996} Suppose $(A,\mathbfit{B}_2,\mathbfit{C}_2)$ in (\ref{equ:G22}) is stabilizable and detectable. Then the closed-loop system in Fig.~\ref{fig:system} is internally stable if and only if so is the system in Fig.~\ref{fig:system:stab}.\hfill$\square$
\end{lem}

%==============================================================================
\begin{figure}[htpb]
\begin{center}
\unitlength=0.7mm
\linethickness{0.4pt}
\begin{picture}(30.00,35.00)
    %\put(7.5,7.5){\dashbox(35,15)[cc]{}}
    \put(10,25){\framebox(10,10)[cc]{$\cP_{22}$}}
    \put(10,10){\framebox(10,10)[cc]{$K$}}
%    \put(15,10){\vector(0,1){10}}
% \put(15,8){\makebox(0,0)[ct]{{$\omega$}}}
    \put(0,30){\vector(1,0){10}}
%    \put(5,15){\line(0,-1){10}}
%    \put(5,5){\line(1,0){40}}
%    \put(45,5){\line(0,1){10}}
 %   \put(50,30){\line(1,0){10}}
 %   \put(35,25){\vector(0,-1){5}}
    %\put(30,18){\vector(-1,0){10}}
    \put(20,30){\line(1,0){10}}
    \put(30,30){\line(0,-1){15}}
    \put(30,15){\vector(-1,0){10}}
    \put(10,15){\line(-1,0){10}}
    \put(0,15){\line(0,1){15}}
%    \put(15,25){\vector(0,-1){5}}
    \put(-2,23){\makebox(0,0)[rc]{{$\breve{\cA}_u$}}}
    %\put(25,18.5){\makebox(0,0)[cb]{{\small$\eta$}}}
%    \put(52,25){\makebox(0,0)[lc]{{$W$}}}
    \put(32,23){\makebox(0,0)[lc]{{$\breve{\cA}_y$}}}
\end{picture}\vskip-5mm
\caption{Equivalent stabilization diagram.}
\label{fig:system:stab}
\end{center}
\end{figure}
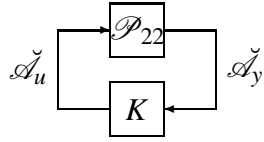
%==============================================================================

%%%%%%%%%%%%%%%%%%%%%%%%%%%%%%%%%%%%%%%%%%%%%%%%%%%%%%%%%%%%%%%%%%%%%%%%%%%%%%%%
%\begin{figure}[htpb]
%\begin{center}
%%\begin{tikzpicture}[scale=0.8, transform shape, node distance=2cm, >=latex']
%%    \node [tmp, name=pinput] (pinput) {};
%%    \node [blockC, right of= pinput] (Plant) {$\cP_{22}$};
%%    \node [output, right of= Plant,node distance=0cm] (pout) {};
%%    \node [Tex, right of= pout , node distance=2.5cm](4){$\breve{\cA}_y$};
%%    \node [blockC, below of=Plant,node distance=2cm] (Controller) {$K$};
%%    \node [output, below of= pinput](cout){};
%%    \node [Tex, left of= pinput,node distance=0.5cm](2){$\breve{\cA}_u$};
%%    \draw [-, ultra thick] (Plant) -| (yout);
%%    \draw [-, ultra thick] (yout) -- (cinput);
%%    \draw [->, ultra thick] (cinput) -- (Controller);
%%    \draw [-, ultra thick] (Controller) -- (cout);
%%    \draw [-, ultra thick] (cout) |- (pinput);
%%    \draw [->, ultra thick] (pinput) -- (Plant);
%%\end{tikzpicture}
%\caption{Equivalent stabilization diagram.}\vskip-3mm
%\label{fig:system:stab}
%\end{center}
%\end{figure}

%%%%%%%%%%%%%%%%%%%%%%%%%%%%%%%%%%%%%%%%%%%%%%%%%%%%%%%%%%%%%%%%%%%%%%%%%%%%%%%%%%%%%%%%%%%%%%%%%%%%
\subsection{Stable Factorization}

%In this subsection, the ring of proper stable rational functions under consideration is defined, which plays a central role in extending the synthesis theory of factorization approach in classical control theory. It is shown that this ring is a proper Euclidean domain, and the Euclidean division process is characterized.
%The entries of the matrix transfer function of a linear quantum system are the elements of $\mC(s)$ which is the field of fractions associated with the set of polynomials $\mC[s]$. Moreover, $\mC[s]$ with the standard definition for the degree of a polynomial is an Euclidean domain, see \cite[Fact A.4.5, p. 126]{vid2011}. Then, under the usual definitions of addition and multiplication in the field $\mC(s)$, $\mS$, the set of all proper bounded-input-bounded-output (BIBO) stable rational functions, as a subset of $\mC(s)$ is a commutative ring with identity, and is a domain.  Consequently, the results of the classical control theory based on the stable factorization approach can be extended to the general case, where the entries of the matrix transfer functions are elements of $\mC(s)$.
Let the transfer function $\cP_{22}$ in (\ref{equ:G22}) have the following coprime factorizations over $\cR \cH_\infty$:
\begin{equation}
\label{fact}
    \cP_{22}
    =
    NM^{-1}
    =
    \wh{M}^{-1}\wh{N},
\end{equation}
where the pairs $(N,M)$ and $(\wh{N}, \wh{M})$ of transfer functions in $\cR \cH_\infty$ specify the right and left factorizations, respectively. Then there exist $U, V, \wh{U}, \wh{V} \in \cR \cH_\infty$ which satisfy the B\'{e}zout identities:
\begin{equation}
        \label{equ:Bezout}
        \wh{V}M-\wh{U}N = I,
        \qquad
        \wh{M}V - \wh{N} U = I.
\end{equation}
The following lemma provides sufficient conditions for the existence of stable coprime factors for the system $\cP_{22}$.

%%%%%%%%%%%%%%%%%%%%%%%%%%%%%%%%%%%%%%%%%%%%%%%%%%%%%%%%%%%%%%%%%%%%%%%%%%%%%%%%%%%%%%%%%%%%%%%%%%%%%%%%%%%%%%%%%%%%%%%%%%%%%%%%%%%%%
\begin{lem} \cite[p. 318]{ZDG_1996}
    Suppose $(A,\mathbfit{B}_2,\mathbfit{C}_2)$ in (\ref{equ:G22}) is stabilizable and detectable. Then the coprime factorizations of $\cP_{22}$ over $\cR \cH_\infty$, described by (\ref{fact}), (\ref{equ:Bezout}), can be chosen so as
    \begin{align}
    	\label{equ:rightcopairs}
        {\small\begin{bmatrix}
               M & U\\
               N  & V
        \end{bmatrix}}
        & =
        {\small\left[
        \begin{array}{c|c c}
              A+\mathbfit{B}_2F & \mathbfit{B}_2 & -L\\
              \hline
                 F & I & 0\\
              \mathbfit{C}_2+\mathbfit{D}_{22}F & \mathbfit{D}_{22} & I
        \end{array}
        \right]},\\
	    \label{equ:leftcopairs}
        {\small\begin{bmatrix}
            \wh{V} & -\wh{U}\\
            -\wh{N}  & \wh{M}
        \end{bmatrix}}
        & =
        {\small\left[
        \begin{array}{c|c c}
              A+L\mathbfit{C}_2 & -(\mathbfit{B}_2+L\mathbfit{D}_{22}) & L\\
              \hline
                 F & I & 0\\
                 \mathbfit{C}_2 & -\mathbfit{D}_{22} & I
        \end{array}
        \right]},
    \end{align}
    where $F \in \mC^{\mu \times 2n}$ and $L \in \mC^{2n \times \mu}$ are such that both matrices $A+\mathbfit{B}_2F$ and $A+L\mathbfit{C}_2$ are Hurwitz. Furthermore, the systems in (\ref{equ:rightcopairs}), (\ref{equ:leftcopairs}) satisfy the general B\'{e}zout identity
\begin{equation}
    \label{equ:GBezout}
    \begin{bmatrix}
        \wh{V} & -\wh{U}\\
        -\wh{N}  & \wh{M}
    \end{bmatrix}
    \begin{bmatrix}
               M & U\\
               N  & V
    \end{bmatrix}
    =
    \begin{bmatrix}
        I & 0\\
        0 & I
    \end{bmatrix}.
\end{equation}
%
%    of Lemma~\ref{lem:GBezout}. % of Appendix~\ref{app:GBI}.
%    \hfill$\square$
\end{lem}
%%%%%%%%%%%%%%%%%%%%%%%%%%%%%%%%%%%%%%%%%%%%%%%%%%%%%%%%%%%%%%%%%%%%%%%%%%%%%%%%%%%%%%%%%%%%%%%%%%%%%%%%%%%%%%%%%%%%%%%%%%%%%%%%%%%%%

%%%%%%%%%%%%%%%%%%%%%%%%%%%%%%%%%%%%%%%%%%%%%%%%%%%%%%%%%%%%%%%%%%%%%%%%%%%%%%%%%%%%%%5
\subsection{The Youla-Ku\v{c}era Parameterization}
%%%%%%%%%%%%%%%%%%%%%%%%%%%%%%%%%%%%%%%%%%%%%%%%%%%%%%%%%%%%%%%%%%%%%%%%%%%%%%%%%%%%%%5

The following lemma, based on \cite[Theorem~12.17, p. 316]{ZDG_1996}, applies the results on the Youla-Ku\v{c}era parameterization in the frequency domain to the closed-loop system being considered.

%%%%%%%%%%%%%%%%%%%%%%%%%%%%%%%%%%%%%%%%%%%%%%%%%%%%%%%%%%%%%%%%%%%%%%%%%%%%%%%%%%%%%%5
\begin{lem}
\label{lem:stab_set} %\cite[Theorem~12.17, p. 316]{ZDG_1996}
Suppose the block $\cP_{22}$ of the modified plant transfer matrix $\cP$ in (\ref{equ:modplnt}) has the coprime factorizations over $\cRH_\infty$, described by (\ref{fact}).
%Then the set of all controllers, which achieve internal stability of the closed-loop system, is parameterized either by
%\begin{equation}
%    \label{equ:ctrl}
%    K=(U + M Q_r)(V + N Q_r)^{-1},
%\end{equation}
%with $Q_r \in \cRH_\infty$ satisfying
%\begin{equation}
%\label{inf1}
%    \det (V + N Q_r)(\infty) \ne 0,
%\end{equation}
%or by
%\begin{equation}
%    \label{equ:ctrr}
%    K=(\wh{V} + Q_{\ell} \wh{N})^{-1}(\wh{U} + Q_{\ell} \wh{M}),
%\end{equation}
%with $Q_{\ell} \in \cRH_\infty$ satisfying
%$
%    \det (\wh{V}+Q_{\ell} \wh{N})(\infty) \ne 0
%$.
Also, let the auxiliary transfer matrices $U, V, \wh{U}, \wh{V} \in \cRH_\infty$ in (\ref{equ:Bezout}) be  chosen so that
$    UV^{-1}=\wh{V}^{-1}\wh{U}
$,
which is equivalent to (\ref{equ:GBezout}). Then the set of all stabilizing controllers
 is parameterized by
\begin{align}
	\nonumber
    K&=(U + MQ)(V + NQ)^{-1}\\
    \label{equ:stab_ctrl:lft}
               &= (\wh{V} + Q \wh{N})^{-1}(\wh{U} + Q \wh{M})
               = {\rm LFT}(O_y, Q),
\end{align}
where the common parameter $Q\in \cRH_{\infty}$ of these factorizations satisfies
\begin{equation}
\label{inf3}
    \det (V+ N Q)(\infty)\ne 0,
\end{equation}
and $O_y:= {\small\begin{bmatrix}
        UV^{-1} & \wh{V}^{-1}\\
        V^{-1}    & -V^{-1}N
    \end{bmatrix}}$ is an auxiliary system.
%     given by
%\begin{equation}
%    \label{equ:O_y}
%    O_y:=
%    {\small\begin{bmatrix}
%        UV^{-1} & \wh{V}^{-1}\\
%        V^{-1}    & -V^{-1}N
%    \end{bmatrix}}.
%\end{equation}
\hfill$\square$
\end{lem}
%%%%%%%%%%%%%%%%%%%%%%%%%%%%%%%%%%%%%%%%%%%%%%%%%%%%%%%%%%%%%%%%%%%%%%%%%%%%%%%%%%%%%%5

In what follows, the class of stabilizing controllers will be parameterized using MFDs. However, they can also be parameterized in the LFT framework due to the relationship between MFD and LFT representations \cite[Lemmas 10.1 and 10.2, pp. 241--242]{ZDG_1996}.

%%%%%%%%%%%%%%%%%%%%%%%%%%%%%%%%%%%%%%%%%%%%%%%%%%%%%%%%%%%%%%%%%%%%%%%%%%%%%%%%%%%%%%5
\section{QUANTUM VERSION OF THE YOULA-KU\v{C}ERA PARAMETERIZATION}\label{sec:QYK}
%%%%%%%%%%%%%%%%%%%%%%%%%%%%%%%%%%%%%%%%%%%%%%%%%%%%%%%%%%%%%%%%%%%%%%%%%%%%%%%%%%%%%%5
We will now employ the material of Sections~\ref{sec:system}--\ref{sec:cntlpara} in order to describe stabilizing coherent quantum controllers in the frequency domain. The following lemma represents $(J_{\mu},J_{\mu})$-unitarity in terms of the Youla-Ku\v{c}era parameter $Q$ from (\ref{equ:stab_ctrl:lft}).
%%%%%%%%%%%%%%%%%%%%%%%%%%%%%%%%%%%%%%%%%%%%%%%%%%%%%%%%%%%%%%%%%%%%%%%%%%%%%%%%%%%%%%%%%%%%%%%%%%%%%%%%%%%%%%%%
\begin{lem}
     \label{JJu:coprime}
     Suppose the controller transfer matrix $K$ is factorized according to (\ref{equ:stab_ctrl:lft}). Then $K$ is $(J_{\mu},J_{\mu})$-unitary if and only if the parameter $Q\in \cRH_{\infty}$ satisfies
     \begin{equation}
         \label{equ:Y1_const}
         \Phi
         +
         Q^{\~}\Lambda
         +
         \Lambda^{\~}Q
         +
         Q^{\~}\Pi Q = 0
     \end{equation}
     for all $s \in \mC$, where
     \begin{align}
     	\label{eq:Phi}
     	\Phi
        &:=
        U^{\~}J_{\mu}U-V^{\~}J_{\mu} V,\\
     	\Lambda
        &:=
        M^{\~}J_{\mu}U-N^{\~}J_{\mu}V,\\
     	\label{eq:Pi}
     	\Pi
        &:=
        M^{\~}J_{\mu} M - N^{\~}J_{\mu} N.
	\end{align}
Furthermore, under the condition (\ref{inf3}), the feedthrough matrix $K(\infty)$ is well-defined and inherits $(J_{\mu}, J_{\mu})$-unitarity from $K$.\end{lem}
%%%%%%%%%%%%%%%%%%%%%%%%%%%%%%%%%%%%%%%%%%%%%%%%%%%%%%%%%%%%%%%%%%%%%%%%%%%%%%%%%%%%%%%%%%%%%%%%%%%%%%%%%%%%%%%%
\begin{proof}
    The $(J_{\mu},J_{\mu})$-unitarity condition for the controller
    \begin{equation}
    \label{KJK}
        K^{\~}(s)J_{\mu}K(s) = J_{\mu},
    \end{equation}
    which must be satisfied for all $s \in \mC$,  is representable in terms of the right factorization from (\ref{equ:stab_ctrl:lft}) as
    \begin{equation}
    \label{JJ}
        \big(
            (U \!+\! M Q)
             (V  \!+\! N Q)^{-1}
        \big)^{\~}
        J_{\mu}
        (U \!+\! M Q)
        (V \!+\! N  Q)^{-1}
        =
        J_{\mu}.
    \end{equation}
    By the properties of system conjugation, (\ref{JJ}) is equivalent to
    $        (U + M  Q)^{\~}
        J_{\mu}
        (U + M Q)
        =
        (V + N  Q)^{\~}
        J_{\mu}
        (V + N Q)
$.
%    \begin{equation}
%        \label{equ:in}
%        (U + M  Q_r)^{\~}
%        J_{\mu}
%        (U + M Q_r)
%        =
%        (V + N  Q_r)^{\~}
%        J_{\mu}
%        (V + N Q_r).
%    \end{equation}
    After regrouping the terms,  the latter equality  takes the form
    \begin{align*}
        U^{\~}J_{\mu}U
        -
        V^{\~}J_{\mu} V
        & +
        Q^{\~}(M^{\~}J_{\mu}U - N^{\~}J_{\mu}V)\\
        & +
        (U^{\~}J_{\mu} M-V^{\~}J_{\mu} N)Q\\
        & +
        Q^{\~}(M^{\~}J_{\mu} M-N^{\~}J_{\mu} N )Q = 0.
    \end{align*}
     This leads to (\ref{equ:Y1_const}), with $\Phi$, $\Lambda$, $\Pi$ given by (\ref{eq:Phi})--(\ref{eq:Pi}). The fact that condition (\ref{inf3}) makes the feedthrough matrix $K(\infty)$ well-defined follows directly from (\ref{equ:stab_ctrl:lft}). The $(J_{\mu},J_{\mu})$-unitarity of $K(\infty)$ is established by taking the limit in (\ref{KJK}) as $s\to \infty$.
\end{proof}
%%%%%%%%%%%%%%%%%%%%%%%%%%%%%%%%%%%%%%%%%%%%%%%%%%%%%%%%%%%%%%%%%%%%%%%%%%%%%%%%%%%%%%%%%%%%%%%%%%%%%%%%%%%%%%%%%%%%%%%%%%%%%%%%%%%%%

The proof of Lemma~\ref{JJu:coprime} shows that the constraint (\ref{equ:Y1_const}) on the Youla-Ku\v{c}era parameter $Q$ inherits its quadratic nature from (\ref{KJK}). However, (\ref{equ:Y1_const}) becomes affine (over the field of reals)  with respect to $Q$ in a particular case when $\Pi=0$. In view of (\ref{fact}), the transfer function $\Pi$ in (\ref{eq:Pi}) is representable as
$
    \Pi = M^{\~}(J_{\mu} - \cP_{22}^{\~}J_{\mu} \cP_{22})M
$,
and hence, it vanishes if the block $\cP_{22}$ of the modified plant is $(J_{\mu}, J_{\mu})$-unitary.
Since $(J_{\mu},J_{\mu})$-unitarity (\ref{KJK}) and its  equivalent dual form $ KJ_{\mu}K^{\~} = J_{\mu} $ (cf. (\ref{equ:JJUnit}) and (\ref{dual})) impose the same constraints on the square transfer matrix $K$, a dual condition to the one described in Lemma~\ref{JJu:coprime} holds for the left factorization of the controller in (\ref{equ:stab_ctrl:lft}). This leads to a dual constraint on $Q$, which corresponds to (\ref{equ:Y1_const}), with $\Phi$, $\Lambda$, $\Pi$ being replaced with their counterparts expressed in terms of $\wh{N}$, $\wh{M}$, $\wh{U}$, $\wh{V}$. %Both constraints can be formulated in a unified fashion.
%%%%%%%%%%%%%%%%%%%%%%%%%%%%%%%%%%%%%%%%%%%%%%%%%%%%%%%%%%%%%%%%%%%%%%%%%%%%%%%%%%%%%%%%%%%%%%%%%%%%%%%%%%%%%%%%%%%%%%%%%%%%%%%%%%%%%
\begin{thm}
\label{thm:stab_ctrl}
Suppose the block $\cP_{22}$ of the modified plant transfer matrix $\cP$ in (\ref{equ:modplnt}) has the coprime factorizations over $\cRH_\infty$ described by (\ref{fact}).  Also, let the transfer matrices $U,V,\wh{U},\wh{V} \in \cRH_\infty$ in (\ref{equ:Bezout}) satisfy the general B{\'e}zout identity  (\ref{equ:GBezout}).
 Then the set of all stabilizing $(J_{\mu},J_{\mu})$-unitary controllers $K$ with a well-defined feedthrough matrix $K(\infty)$ is parameterized by (\ref{equ:stab_ctrl:lft}), where the parameter $Q$ belongs to the set
%\begin{align}
%    \nonumber
%    K&=(U + MQ)(V + NQ)^{-1}\\
%    \label{equ:ctrl:thm}
%                 &=(\wh{V} + Q \wh{N})^{-1}(\wh{U} + Q \wh{M}),%\\
%                 %&= {\rm LFT}(O_y, Q)
% \end{align}
% where
 \begin{equation}
 \label{cQ}
    \cQ:=
    \big\{
        Q\in \cRH_{\infty}\
        {\rm satisfying}\ (\ref{inf3})\ {\rm and}\ (\ref{equ:Y1_const})
    \big\}.
 \end{equation}
% \begin{equation}
%    \label{HO_set:n}
%    Q \in \cQ :=
%    \left \{
%        \begin{array}{c|c}
%            Q_r \in \cRH_\infty &
%            \begin{array}{c}
%            Q_r ~\textrm{satisfies (\ref{equ:Y1_const})}\\
%        	\textrm{for all}~s \in \mC \textrm{, and}\\
%        	\det (V + N Q)(\infty)\neq 0
%            \end{array}
%        \end{array}
%    \right \}.
%\end{equation}
\end{thm}
%%%%%%%%%%%%%%%%%%%%%%%%%%%%%%%%%%%%%%%%%%%%%%%%%%%%%%%%%%%%%%%%%%%%%%%%%%%%%%%%%%%%%%%%%%%%%%%%%%%%%%%%%%%%%%%%%%%%%%%%%%%%%%%%%%%%%

\begin{proof}
This theorem is proved by combining Lemmas~\ref{lem:stab_set} and~\ref{JJu:coprime}. Indeed, since the underlying coprime factorizations are assumed to satisfy the general B\'{e}zout identity (\ref{equ:GBezout}), then (\ref{equ:Y1_const})  can be applied to the common parameter $Q$ in (\ref{equ:stab_ctrl:lft}) in order to describe all stabilizing $(J_{\mu},J_{\mu})$-unitary controllers $K$. Their feedthrough matrices $K(\infty)$ are well-defined provided the additional condition (\ref{inf3}) is also satisfied. The resulting class of admissible $Q$ is given by (\ref{cQ}).
\end{proof}
%%%%%%%%%%%%%%%%%%%%%%%%%%%%%%%%%%%%%%%%%%%%%%%%%%%%%%%%%%%%%%%%%%%%%%%%%%%%%%%%%%%%%%%%%%%%%%%%%%%%%%%%%%%%%%%%%%%%%%%%%%%%%%%%%%%%%

Theorem~\ref{thm:stab_ctrl} provides a frequency domain parameterization of all stabilizing $(J_{\mu},J_{\mu})$-unitary controllers with a well-defined feedthrough matrix and leads to the following theorem.
\begin{thm}
\label{thm:stab_ctrl:coherent}
Under the assumptions of Theorem~\ref{thm:stab_ctrl}, the MFDs (\ref{equ:stab_ctrl:lft}) describe a set of stabilizing PR quantum controllers $K$, where the parameter $Q$ belongs to the following class $\wh{\cQ}$ defined in terms of (\ref{cQ}) and (\ref{mD}):
%\begin{align}
%    \label{equ:ctrl:thm:coherent:1}
%        K(s)
%        &=
%            (U + MQ)(V + NQ)^{-1}\\
%    \label{equ:ctrl:thm:coherent:2}
%        &=
%            (\wh{V} + Q \wh{N})^{-1}(\wh{U} + Q \wh{M})%\\
%                 %&= {\rm LFT}(O_y, Q),
%\end{align}
%where
\begin{align}
\nonumber
    \wh{\cQ}
    :=
    \big\{
        Q\in \cQ:\ &
        K\ {\rm in}\
        (\ref{equ:stab_ctrl:lft})\
        {\rm is\ spectrally\ generic,}
        \\
 \label{cQhat}
        & {\rm and}\ K(\infty)\in \mD_{\mu}
    \big\}.
 \end{align}
 %\begin{equation}
%    \label{HO_set:c}
%    Q \in \wh{\cQ} :=
%    \left \{
%        \begin{array}{c|c}
%            Q \in \cQ &
%            \begin{array}{c}
%                s_j +\overline{s}_k \neq 0 ~\textrm{for all } \\
%                s_j, s_k \in \bP(V + N Q),\\
%                G(S_{\rm C},\infty)=0\\
%                \textrm{for a unitary}~S_{\rm C}
%            \end{array}
%        \end{array}
%    \right \}.
%\end{equation}
%Here, the transfer matrices $U$, $V$, $\wh{U}$ and $\wh{V} \in \cRH_\infty$ are chosen so as to satisfy the general B{\'e}zout identity of Appendix~\ref{app:GBI}, $\cQ$ and $\bP(\cdot)$ are in turn defined in (\ref{HO_set:n}) and (\ref{equ:def:poles}), and
%\begin{equation*}
%    G(S,s) :=U + MQ-\Delta(S,0)(V+ NQ).
%\end{equation*}
 \end{thm}
\begin{proof}
The assertion of the theorem is established by combining Theorem~\ref{thm:stab_ctrl} with the frequency domain criterion of PR provided by Lemma~\ref{lem:PR_Freq}.
%, for parameterization of all the stabilizing $(J_{\mu},J_{\mu})$-unitary controllers, and Lemma~\ref{lem:poles}. Since it is assumed that the MFDs of $\Gamma_{\rm C}$ in (\ref{equ:ctrl:thm:coherent:1}) and (\ref{equ:ctrl:thm:coherent:2}) are irreducible, any state space realization of the MFDs with order equal to the degree of the determinant of the denominator matrix will be a minimal realization; see \cite[Theorem 6.5-1, p. 439]{KTh_1980}. Then, (\ref{equ:ctrl:thm:coherent:1}) and (\ref{equ:ctrl:thm:coherent:2}) are the right and the left coprime factorizations of the controller. Based on the results of Lemma~\ref{lem:poles}, the condition in (\ref{equ:pole_state}) for the controller is considered in the set $\wh{\cQ}$ for the associated parameter $Q$. Moreover, the constraint on the feedthrough matrix $K(\infty)=\Delta(S_{\rm C},0)$ in Lemma~\ref{lem:PR_Freq} is taken into account by $G(S_{\rm C},\infty)=0$ for a unitary $S_{\rm C}$ in (\ref{HO_set:c}).
\end{proof}
%%%%%%%%%%%%%%%%%%%%%%%%%%%%%%%%%%%%%%%%%%%%%%%%%%%%%%%%%%%%%%%%%%%%%%%%%%%%%%%%%%%%%%%%%%%%%%%%%%

%The result of Corollary~\ref{thm:stab_ctrl:coherent} provides a parameterization of stabilizing linear  coherent quantum controllers in the frequency domain.

Theorem~\ref{thm:stab_ctrl:coherent}
parameterizes a \emph{subset} of stabilizing PR quantum controllers in  the frequency domain through the representation (\ref{equ:stab_ctrl:lft}) and the set $\wh{\cQ}$ in (\ref{cQhat}). This subset of controllers does not exhaust all stabilizing coherent quantum controllers. However, the  discrepancy between these two classes of controllers is only caused by the technical condition of spectral genericity which comes from Lemma~\ref{lem:PR_Freq}.
%In line with the latter condition, recall that the set $\bP(\Gamma)$ of poles of a transfer function $\Gamma$, with a minimal state-space realization (\ref{equ:hoc_freq}), coincides with the spectrum $\sigma(A)$ of the matrix $A$. On the other hand, this set is representable in terms of the coprime factorizations
%\begin{equation}
%\label{SigmaXi}
%    \Gamma = \Sigma\Xi^{-1} = \wh{\Xi}^{-1}\wh{\Sigma}
%\end{equation}
%as
%\begin{equation}
%\label{PZ}
%    \sigma(A)= \bP(\Sigma) \bigcup \bZ(\Xi) = \bP(\wh{\Sigma}) \bigcup \bZ(\wh{\Xi}),
%\end{equation}
%where $\bZ(f):= \{s\in \mC:\ \det f(s)=0\}$ for square transfer matrices $f$. Therefore, the property of being spectrally generic in the sense of Definition~\ref{def:gen}  can be equivalently formulated  for $\Gamma$ in (\ref{SigmaXi}) in terms of the poles and ``zeros'' of its factors in (\ref{PZ}).

%%%%%%%%%%%%%%%%%%%%%%%%%%%%%%%%%%%%%%%%%%%%%%%%%%%%%%%%%%%%%%%%%%%%%%%%%%%%%%%%%%%%%%%%%%%%%%%%%%
\section{COHERENT  QUANTUM WEIGHTED $\mathcal{H}_2$ AND $\mathcal{H}_\infty$ CONTROL PROBLEMS IN THE FREQUENCY DOMAIN}\label{sec:H2}
%%%%%%%%%%%%%%%%%%%%%%%%%%%%%%%%%%%%%%%%%%%%%%%%%%%%%%%%%%%%%%%%%%%%%%%%%%%%%%%%%%%%%%%%%%%%%%%%%%
The following lemma, which is  given here for completeness, employs the factorization approach in order to obtain a more convenient representation of the closed-loop transfer function.

%%%%%%%%%%%%%%%%%%%%%%%%%%%%%%%%%%%%%%%%%%%%%%%%%%%%%%%%%%%%%%%%%%%%%%%%%%%%%%%%%%%%%%%%%%%%%%%%%%
\begin{lem}
    \label{lem:closed-loop}
    Under the assumptions of Theorem~\ref{thm:stab_ctrl}, for any stabilizing  controller $K$ parameterized by  (\ref{equ:stab_ctrl:lft}),   the corresponding closed-loop transfer matrix $G$ in (\ref{G}) is representable as
    \begin{equation}
    \label{GTTT}
        G
        =
        T_0 + T_1 Q T_2,
    \end{equation}
    where
    \begin{align}
    		\label{eq:T_0_eq:T_2}
    		\!\!\!\!\!\!T_0
                \!:=\!
                    \cP_{11}\! +\! \cP_{12}U\wh{M}\cP_{21},\quad
    		T_1
                \!:= \!
                    \cP_{12} M,\quad
    		%\label{eq:T_2}
    		T_2
                \!:=\!
                    \wh{M}\cP_{21}.\!\!\!\!
    	\end{align}
\end{lem}
%%%%%%%%%%%%%%%%%%%%%%%%%%%%%%%%%%%%%%%%%%%%%%%%%%%%%%%%%%%%%%%%%%%%%%%%%%%%%%%%%%%%%%%
%\begin{proof}
%    By substituting $\cP_{22}$ from (\ref{fact}) and $K$ from (\ref{equ:stab_ctrl:lft}) into (\ref{G}), it follows that
%    \begin{align*}
%        G = &  \cP_{11} +  \cP_{12}(U + MQ)(V + NQ)^{-1}\\
%        &\x (I-NM^{-1}(U + MQ)(V + NQ)^{-1})^{-1}\cP_{21}\\
%        =& \cP_{11} + \cP_{12}(U + MQ)
%        (V \!+\! NQ\!-\!NM^{-1}(U \!+\! MQ))^{-1}\cP_{21}\\
%        =& \cP_{11} + \cP_{12}(U + MQ)
%        (V -NM^{-1}U)^{-1}\cP_{21}\\
%        =& \cP_{11} + \cP_{12}(U + MQ)\wh{M}\cP_{21},
%    \end{align*}
%    which leads to (\ref{GTTT}) and (\ref{eq:T_0_eq:T_2}).
%    Here, use is made of the relations
%    $
%        V -NM^{-1}U = V - \wh{M}^{-1}\wh{N}U = V - \wh{M}^{-1}(\wh{M}V-I) = \wh{M}^{-1}
%    $
%    which are obtained from (\ref{fact}) and (\ref{equ:Bezout}).
%\end{proof}
%%%%%%%%%%%%%%%%%%%%%%%%%%%%%%%%%%%%%%%%%%%%%%%%%%%%%%%%%%%%%%%%%%%%%%%%%%%%%%%%%%%%%%%

Lemma~\ref{lem:closed-loop} allows the following coherent quantum weighted $\mathcal{H}_2$ and $\mathcal{H}_\infty$ control problems to be formulated in the frequency domain.

%%%%%%%%%%%%%%%%%%%%%%%%%%%%%%%%%%%%%%%%%%%%%%%%%%%%%%%%%%%%%%%%%%%%%%%%%%%%%%%%%%%%%%
\subsection{Coherent Quantum Weighted $\cH_2$ Control Problem}
%%%%%%%%%%%%%%%%%%%%%%%%%%%%%%%%%%%%%%%%%%%%%%%%%%%%%%%%%%%%%%%%%%%%%%%%%%%%%%%%%%%%%%

Using the representation (\ref{GTTT}), we formulate a coherent quantum weighted $\cH_2$ control problem as the constrained minimization problem
    \begin{align}
    	\label{H2}
        E:= \|W_{\rm out}G W_{\rm in}\|_2^2
        = \|\mbT_0 + \mbT_1Q\mbT_2\|_2^2
        \longrightarrow
        \min
    \end{align}
with respect to $Q \in \wh{\cQ}$, where the set $\wh{\cQ}$ is given by (\ref{cQhat}). Here,
\begin{equation}
	\label{eq:bT_0_eq:bT_2}
	\mbT_0
    :=
    W_{\rm out}T_0W_{\rm in},
    \qquad
	\mbT_1
    :=
    W_{\rm out}T_1,
    \qquad
	\mbT_2:=T_2W_{\rm in},
\end{equation}
where $T_0$, $T_1$, $T_2$ are defined by (\ref{eq:T_0_eq:T_2}). Also, $W_{\rm in}, W_{\rm out}\in \cRH_{\infty}$ are given strictly proper weighting transfer functions for the closed-loop system $G$ which ensure that $\mbT_0 + \mbT_1Q\mbT_2 \in \cH_2$.
    The $\cH_2$-norm $\|\cdot \|_2$ is associated with the inner product
$
    \bra \Gamma_1,\Gamma_2 \ket:=\frac{1}{2\pi}\int_{-\infty}^{+\infty} \bra \Gamma_1(i\omega),\Gamma_2(i\omega) \ket_{\rm F} \rd \omega
$.
By using the standard properties of inner products in complex Hilbert spaces \cite{reed1980},  the cost functional $E$ in (\ref{H2}) can be represented as
\begin{equation}
\label{prob:cost:H2}
    E
%    & =
%    \|\mbT_0\|_2^2
%    +
%    2\Re \bra \mbT_0 , \mbT_1Q\mbT_2 \ket
%    +
%    \| \mbT_1 Q \mbT_2\|_2^2\\
    =
    \|\mbT_0\|_2^2
    +
    2\Re \bra \wh{\mbT}_0, Q\ket
    +
    \bra
        Q,
        \wh{\mbT}_1 Q \wh{\mbT}_2
    \ket,
\end{equation}
where
\begin{equation}
\label{T0_T2}
	\wh{\mbT}_0 :=\mbT_1^{\~}\mbT_0\mbT_2^{\~},
\qquad
	\wh{\mbT}_1:=\mbT_1^{\~}\mbT_1,
\qquad
	\wh{\mbT}_2:=\mbT_2\mbT_2^{\~}.
\end{equation}
In comparison to the original coherent quantum LQG control problem \cite{NJP_2009}, the coherent quantum weighted $\mathcal{H}_2$ control problem (\ref{H2}) allows for considering the cost of the unavoidable quantum noise fed through the plant by the controller.

%%%%%%%%%%%%%%%%%%%%%%%%%%%%%%%%%%%%%%%%%%%%%%%%%%%%%%%%%%%%%%%%%%%%%%%%%%%%%%%%%%%%%%
\subsection{ Coherent Quantum Weighted $\cH_\infty$ Control Problem}
%%%%%%%%%%%%%%%%%%%%%%%%%%%%%%%%%%%%%%%%%%%%%%%%%%%%%%%%%%%%%%%%%%%%%%%%%%%%%%%%%%%%%%

Similarly to (\ref{H2}),  a coherent quantum weighted $\mathcal{H}_\infty$ control problem is formulated as the constrained minimization problem
  \begin{equation}
  \label{Hinf}
        \|G\|_\infty
        =
        \|\mbT_0 + \mbT_1Q\mbT_2\|_\infty
        \longrightarrow
        \min
  \end{equation}
with respect to $Q\in \wh{\cQ}$, where the set $\wh{\cQ}$ is defined by (\ref{cQhat}). Here, $\mbT_0$, $\mbT_1$ and $\mbT_2$ are given by  (\ref{eq:bT_0_eq:bT_2}),  where,  this time,  the weighting transfer functions   $W_{\rm in}, W_{\rm out} \in \cRH_{\infty}$ are not necessarily strictly proper. Recall that the norm in the Hardy space $\cH_\infty$ is defined by
     $
        \|\Gamma\|_\infty
        :=
        \sup_{\omega \in \mR}
        \sigma_{\max}(\Gamma(i\omega))
     $,
where $\sigma_{\max}(\cdot)$ denotes the largest singular value of a matrix.
Note that both problems (\ref{H2}) and (\ref{Hinf}) are organised as constrained versions of the model matching problem \cite{francis87}. Since the $\cH_2$ control problem is based on a Hilbert space norm, its solution can be approached by using a variational method in the frequency domain, which employs differentiation of the cost $E$ with respect to the Youla-Ku\v{c}era parameter $Q$ and is qualitatively different from the state-space techniques of \cite{VP_2013a}.

%%%%%%%%%%%%%%%%%%%%%%%%%%%%%%%%%%%%%%%%%%%%%%%%%%%%%%%%%%%%%%%%%%%%%%%%%%%%%%%%
\section{PROJECTED GRADIENT DESCENT SCHEME FOR THE COHERENT  QUANTUM WEIGHTED $\cH_2$ CONTROL PROBLEM}\label{sec:PGS}
%%%%%%%%%%%%%%%%%%%%%%%%%%%%%%%%%%%%%%%%%%%%%%%%%%%%%%%%%%%%%%%%%%%%%%%%%%%%%%%%
Suppose the set $\wh{\cQ}$ in (\ref{cQhat}) is nonempty, and hence, there exist stabilizing PR quantum controllers for a given quantum plant. %Tractable conditions for the existence of such controllers remain an open problem which is not considered here.
By using the representation (\ref{prob:cost:H2}) and regarding the transfer function $Q \in \cRH_{\infty}$ as an independent optimization variable, it follows that the first variation of the cost functional  $E$ in (\ref{H2}) with respect to $Q$ can be computed as
\begin{equation}
\label{Egrad}
\delta E
    =
    \Re
    \bra
        \nabla E,
        \delta Q
    \ket,
    \qquad
    \nabla E
    := 2(\wh{\mbT}_0+\wh{\mbT}_1 Q \wh{\mbT}_2),
\end{equation}
where use is also made of (\ref{T0_T2}). In order to yield a PR quantum controller, $Q$ must satisfy the constraint (\ref{equ:Y1_const}) whose variation leads to
\begin{equation}
    \label{equ:Y2_const:var}
    \delta Q^{\~}(\Lambda+\Pi Q)+(\Lambda^{\~}+Q^{\~}\Pi)\delta Q=0.
\end{equation}
In view of the uniqueness theorem for analytic functions \cite[pp. 369--371]{M65}, the resulting constrained optimization problem can be reduced to that for purely imaginary $s = i\omega$, with $\omega\in \mR$. The transfer matrices $\delta Q$, satisfying (\ref{equ:Y2_const:var})  at frequencies  $\omega$ from a given set $\Omega\subset\mR$, form a real subspace of transfer functions
\begin{equation}
\label{cS}
    \!\cS
    \!\!:=\!
    \big\{\!
        X\!\in\!  \cRH_{\infty}\!:
        (X^*(\Lambda\!+\!\Pi Q)\!+\!(\Lambda^*\!\!+\!Q^*\Pi)X)\big|_{i\Omega}\!=\!0\!
    \big\}.\!\!\!\!\!\!\!\!\!
\end{equation}
For practical purposes, the set $\Omega$ is used to ``discretize'' the common frequency range of the given weighting transfer functions $W_{\rm in}$,  $W_{\rm out}$ in the coherent quantum weighted $\cH_2$ control problem (\ref{H2}). A numerical solution of  this problem can be implemented in the form of the following projected gradient descent scheme for finding a critical point of the cost functional $E$ with respect to $Q$ subject to (\ref{equ:Y1_const}) at a finite set of frequencies $\Omega$:
\begin{enumerate}
  \item initialize $Q\in \cRH_{\infty}$ so as to satisfy (\ref{equ:Y1_const}), which yields a stabilizing PR quantum controller; %Choose the input and output weights $W_{\rm in}(i\omega)$ and $W_{\rm out}(i\omega)$, and assign a discrete frequency array $\omega$ in the set $[-\omega_{\ell},\omega_u]\subseteq \mR$. Set the step size $\alpha >0$.
  \item calculate $\nabla E(i\omega)$ according to (\ref{Egrad}) for each frequency $\omega \in \Omega$;
  \item compute $\delta Q(i\omega) = -\alpha \mathrm{Proj}_{\cS}(\nabla E(i\omega))$ by using a projection onto the set  $\cS$ and a parameter $\alpha>0$;
  \item update $Q$ to $Q+\delta Q$, and go to the second step.
\end{enumerate}
The gradient projection $\mathrm{Proj}_{\cS}(\nabla E)$ onto  the set $\cS$ in (\ref{cS}) is computed in the third step of the algorithm by solving a convex optimization problem
on a Hilbert space with the direct sum of the Frobenius inner products of the projection errors at frequencies $\omega \in \Omega$.
This computation, which will be discussed elsewhere,  also involves the interpolation of transfer functions; see \cite{B90,hel98} for more details. The discrete frequency set $\Omega$ and the step-size parameter $\alpha$ can be chosen adaptively at each iteration of the algorithm. The outcome of the algorithm is considered to be acceptable if $Q$ belongs to the set $\wh{\cQ}$ defined by (\ref{cQhat}) of Theorem~\ref{thm:stab_ctrl:coherent}. In the case when $Q$ satisfies all the conditions in (\ref{cQhat}) except for the spectral genericity,  slightly different weighting matrices can be used in order to remedy the situation.
%%%%%%%%%%%%%%%%%%%%%%%%%%%%%%%%%%%%%%%%%%%%%%%%%%%%%%%%%%%%%%%%%%%%%%%%%%%%%%%%%%%%%
\section{CONCLUSION}\label{sec:Conclusion}
%%%%%%%%%%%%%%%%%%%%%%%%%%%%%%%%%%%%%%%%%%%%%%%%%%%%%%%%%%%%%%%%%%%%%%%%%%%%%%%%%%%%%%

The set of stabilizing linear coherent quantum controllers for a given linear quantum plant has been parameterized using a Youla-Ku\v{c}era factorization  approach. This approach has provided a formulation of coherent quantum weighted $\cH_2$ and $\cH_\infty$ control problems for linear quantum systems in the frequency domain. These problems resemble constrained versions of the classical model matching problem. A projected gradient descent scheme has been outlined for numerical solution of the coherent quantum weighted $\cH_2$ control problem in the frequency domain. The proposed framework can also be used to develop tractable conditions for the existence of stabilizing quantum controllers for a given quantum plant, which remains an open problem. This is a subject of future research and will be considered in subsequent publications.

\end{document}